\documentclass[11pt]{article}
\usepackage{fullpage}
\usepackage{authblk}
\usepackage{amsthm}
\usepackage{amsmath}
\usepackage[table,xcdraw]{xcolor}
\usepackage{multirow}
\usepackage{wrapfig}
\usepackage[leftcaption]{sidecap}
\usepackage{multicol}

\usepackage{thm-restate}
\usepackage{nicefrac}
\usepackage[boxed,lined,linesnumbered, ruled]{algorithm2e}

\usepackage{comment,amsfonts,amsmath,amsthm,graphicx,subcaption}

\usepackage{natbib}

\definecolor{Darkblue}{rgb}{0,0,0.4}
\definecolor{Brown}{cmyk}{0,0.61,1.,0.60}
\definecolor{Purple}{cmyk}{0.45,0.86,0,0}
\definecolor{Darkgreen}{rgb}{0.133,0.543,0.133}

\usepackage[colorlinks,linkcolor=Darkblue,filecolor=blue,citecolor=blue,urlcolor=Darkblue,pagebackref]{hyperref}
\usepackage[nameinlink]{cleveref}

\newcommand{\commentout}[1]{}

\def\epsilon{\varepsilon}
\def\eps{\varepsilon}

\newcommand{\alert}[1]{\textbf{\color{red}
		[[[#1]]]}\marginpar{\textbf{\color{red}**}}\typeout{ALERT:\@
		\the\inputlineno: #1}}

\usepackage[colorinlistoftodos,prependcaption,textsize=tiny]{todonotes}
\newcommand{\atodo}[1]{}
\newcommand{\atodoin}[1]{}

\newcommand{\gtodo}[1]{}
\newcommand{\gtodoin}[1]{}

\def\eps{\epsilon}

\newcommand{\cost}{{\rm cost}}

\newcommand{\opt}{{\rm OPT}}

\newcommand{\poly}{{\rm poly}}

\newcommand{\mommit}[1]{}
\newcommand{\namedref}[2]{\hyperref[#2]{#1~\ref*{#2}}}

\newcommand{\ball}{{\sf ball}}

\newcommand{\SC}{\texttt{SETTLE-CLUSTER}\xspace}

\newcommand{\guess}{\emph{\texttt{guess}}\xspace}

\theoremstyle{plain}
\newtheorem{theorem}{Theorem}[]
\newtheorem{lemma}{Lemma}[]

\newtheorem{remark}[lemma]{Remark}
\newtheorem{definition}{Definition}

\newtheorem{assumption}{Assumption}
\newtheorem{observation}{Observation}

\newcommand{\cB}{\mathcal{B}}

\newcommand{\cD}{\mathcal{D}}
\newcommand{\cE}{\mathcal{E}}
\newcommand{\cF}{\mathcal{F}}

\newcommand{\cI}{\mathcal{I}}
\newcommand{\cJ}{\mathcal{J}}

\newcommand{\cL}{\mathcal{L}}

\newcommand{\cS}{\mathcal{S}}

\newcommand{\cX}{\mathcal{X}}

\newcommand{\kcen}{\textsc{$k$-Center}\xspace}

\newcommand{\pp}{\textsf{P}\xspace}
\newcommand{\np}{\textsf{NP}\xspace}
\newcommand{\fpt}{\textsf{FPT}}

\newcommand{\nonblind}[1]{}

\newcommand{\fairsorout}{Fair SoR with Outliers\xspace}
\newcommand{\colkceno}{Colorful SoR with Outliers\xspace}
\newcommand{\fairsoroutsc}{fair SoR with outliers\xspace}
\newcommand{\colkcenosc}{colorful SoR with outliers\xspace}
\newcommand{\fairnoroutsc}{fair norm of radii with outliers\xspace}

\newcommand{\fairksupo}{Fair Sup-SoR with Outliers\xspace}
\newcommand{\fairksuposc}{fair sup-SoR with outliers\xspace}

\newcommand{\colksupo}{Colorful SoR with Outliers\xspace}
\newcommand{\colksuposc}{colorful SoR with outliers\xspace}
\newcommand{\met}{\ensuremath{(X,d)}\xspace}

\newcommand{\mets}{\ensuremath{(F\cup C,d')}\xspace}
\newcommand{\points}{\ensuremath{X}}
\newcommand{\outcnt}{\ensuremath{z}\xspace}
\newcommand{\groups}{\ensuremath{\mathcal{X}}}

\newcommand{\foins}{\ensuremath{(\met,k,z,\groups=\{X_i\}_{i \in [t]},\{k_i\}_{i \in [t]})}}
\newcommand{\fosupins}{\ensuremath{(\mets,k,z,\cF=\{F_i\}_{F_i \in [\ell]}})}
\newcommand{\coins}{\ensuremath{(\mets,k,z,\cF=\{F_i\}_{i \in [k]})}}
\newcommand{\coinsp}{\ensuremath{(\mets,k,z,\cF'=\{F'_i\}_{i \in [k]})}}

\title{FPT Approximations for Fair Sum of Radii with Outliers and General Norm Objectives}

\author{Ameet Gadekar}

\affil{CISPA Helmholtz Center for Information Security, Saarbr\"{u}cken, Germany.\qquad\qquad\qquad Email:  \texttt{ameet.gadekar@cispa.de}}
 \date{}
\begin{document}
\pagenumbering{gobble}

\maketitle



\begin{abstract}

The \emph{sum of radii} problem is a classical clustering problem in which, given a set $X$ of points and an integer $k$, the goal is to place $k$ balls that cover $X$ while minimizing the sum of their radii.
Recent work has focused on incorporating modern constraints such as fairness and robustness, motivated by biased and noisy data.
We study the \emph{fair sum of radii with outliers} problem, where the chosen centers must satisfy group-based representation constraints while allowing up to $z$ points to be excluded.

We present a $(3+\eps)$-approximation algorithm that runs in fixed-parameter tractable time parameterized by $k$.
Our framework extends to the more general setting where the objective is a monotone symmetric norm of the radii, achieving a $(3+\eps)$-approximation for any fixed norm; this guarantee is tight under Gap-ETH.
Moreover, the algorithm is \emph{oblivious} to the choice of norm: it outputs a small list of candidate solutions such that, for every monotone symmetric norm $f$, the list contains a $(3+\eps)$-approximate solution under $f$.

Our approach is based on a novel iterative ball-finding framework that uncovers a structural trichotomy in the optimal clustering, enabling us to directly construct fair solutions while handling outliers.
Finally, we extend our techniques to the more general \emph{fair-range} setting, where each group is subject to both lower and upper bounds.
\end{abstract}

    	\vfill
	\begin{multicols}{2}
		{
			\setcounter{secnumdepth}{5}
			\setcounter{tocdepth}{2} \tableofcontents
		}
	\end{multicols}
	\newpage

\pagenumbering{arabic}

\section{Introduction}

Clustering is a fundamental primitive for analyzing data, with applications in optimization, operations research, computational geometry, and machine learning. Classical objectives such as $k$-median, $k$-means, and $k$-center have been extensively studied for more than four decades. Among these, $k$-center is perhaps the simplest and most well-understood problem~\cite{gonzalez1985clustering,hochbaum1985best,DBLP:journals/heuristics/Garcia-DiazSMM17}. Given a metric space $(X,d)$ and an integer $k$, the goal is to place $k$ balls to cover all points while minimizing the maximum radius. The problem admits a simple greedy $2$-approximation algorithm, which is optimal assuming $\pp \neq \np$.

A closely related objective is the \emph{sum of radii} problem~\cite{MS89,CHARIKAR2004417}. In this problem, we are given a metric space $(X,d)$ and an integer $k$, and the goal is to select centers $S=\{s_1,\dots,s_k\} \subseteq X$ and radii $R=\{r_1,\dots,r_k\}$ such that the union of the balls $\bigcup_{i \in [k]} \ball(s_i,r_i)$ covers all points in $X$, while minimizing $\sum_{i \in [k]} r_i$. Compared to $k$-center, this objective better preserves the structure of the underlying clusters by penalizing uneven cluster sizes.

From a computational perspective, the sum of radii problem behaves very differently from \kcen. On the one hand, it admits a QPTAS~\cite{10.5555/1347082.1347172}, suggesting the possibility of a PTAS, in contrast to \kcen, which is APX-hard even in Euclidean space~\cite{awasthi2015hardness}. On the other hand, even when the centers are fixed, computing the optimal assignment of points (i.e., determining the optimal radii) is \np-hard~\cite{DBLP:journals/corr/DHKS23}, unlike \kcen where assignments are trivial. 
This contrast highlights the richer structure and additional algorithmic challenges inherent in the sum of radii objective.

\paragraph*{Fairness and Robustness.}
Two central directions in modern clustering are robustness to outliers and fairness constraints on the solution.
Classical objectives such as \kcen and sum of radii are highly sensitive to perturbations: even a small number of adversarial or extreme points can significantly distort the solution.
The robust variant addresses this by allowing the algorithm to discard up to $z$ points (called outliers), thereby stabilizing the objective~\cite{CHARIKAR2004417,DBLP:journals/talg/ChakrabartyN19,BuchemERW24}. 
Independently, in \emph{fair clustering}, points are associated with sensitive attributes (such as gender, expertise, or geographic region), and the goal is to ensure that the clustering solution satisfies prescribed fairness constraints across these groups, such as balanced or proportional representation.

This direction has received significant attention over the past decade~\cite{DBLP:conf/nips/Chierichetti0LV17,DBLP:conf/icml/JonesNN20,DBLP:conf/isaac/CartaDHR024,DBLP:conf/approx/BandyapadhyayC25,DBLP:conf/www/GadekarGT25,DBLP:conf/icalp/Rosner018, DBLP:conf/icml/KleindessnerAM19, DBLP:conf/kdd/ThejaswiGOO22, 10.1145/3292500.3330987, DBLP:conf/icalp/Micha020, DBLP:journals/corr/abs-2002-07892, DBLP:conf/approx/Bercea0KKRS019,DBLP:conf/nips/BeraCFN19}. 
While both robustness and fairness have been extensively studied in isolation, their combination introduces substantial new challenges, and techniques for either setting do not directly extend. 
These directions are further motivated by the presence of noise and bias in modern datasets, where both robustness and fairness considerations naturally arise.

\paragraph*{Our Problem.}
Motivated by these challenges, we study the \emph{fair sum of radii with outliers} problem, where both fairness and robustness constraints are present simultaneously. 
The input consists of a metric space $(X,d)$ on $n$ points, an integer $k$, an outlier budget $z$, and a partition of $X$ into groups $X_1,\dots,X_t$ with upper bounds $k_i$ on the number of centers chosen from each group. 
The goal is to select $k$ centers and radii such that:
(i) at least $n-z$ points are covered,
(ii) $|S \cap X_i| \le k_i$ for all $i$, and
(iii) the sum of radii is minimized.
Thus, fairness is enforced through constraints on the selection of centers across groups, while robustness is handled by allowing up to $z$ uncovered points (outliers).
A similar combination of fairness and outliers has been studied for the $k$-center objective~\cite{DBLP:journals/talg/HarrisPST19,DBLP:conf/sosa/AneggKZ23}, including more general formulations such as robust matroid $k$-center.
In contrast, for the sum of radii objective, this combination has not been studied to the best of our knowledge.

The sum of radii problem has been studied separately under robustness and fairness constraints. 
For the outlier version, the best known polynomial-time approximation is $(3+\eps)$ due to~\cite{BuchemERW24}, while for the fairness notion considered in this paper (i.e., constraints on center selection), the best known factor is also $(3+\eps)$~\cite{ChenXXZ24}, albeit in \fpt\ time parameterized by $k$. 
When the underlying metric has doubling dimension $d$, $(1+\eps)$-approximations are known for both variants separately~\cite{DBLP:conf/aaai/BanerjeeBGH25}. 
These results align with a broader line of work on \fpt-approximation algorithms for clustering under fairness and outliers~\cite{DBLP:conf/isaac/CartaDHR024,DBLP:conf/approx/BandyapadhyayC25,ChenXXZ24,DBLP:journals/corr/LGXZ25}.

\paragraph*{Beyond Sum: Norm Objective.}
While our primary focus is on the sum of radii objective, our framework naturally extends to monotone symmetric norms over the radii.
Recent work has focused on designing approximation algorithms for such general objectives, rather than tailoring algorithms to specific ones such as sum or max~\cite{chakrabarty2019approximation, DBLP:conf/soda/ChlamtacMV22, AbbasiClustering23, DBLP:conf/soda/HeroldKS25,DBLP:conf/soda/HeroldKS26}.
These objectives provide a unified framework that captures and interpolates between several classical clustering objectives. 
In particular, monotone symmetric norms over the radii captures $k$-center ($\ell_\infty$), sum of radii ($\ell_1$), and sum of squared radii ($\ell_2$)~\cite{alt2006minimum, liu2022primal}.
Motivated by this direction, we study fair sum of radii with outliers under such norm objectives.

Our main result is a $(3+\eps)$-approximation algorithm in \fpt\ time for fair sum of radii with outliers. Moreover, our framework extends to this more general norm setting, yielding a $(3+\eps)$-approximation for any fixed monotone symmetric norm.

\subsection{Our Results}

We begin with our result for the sum of radii objective, which forms the conceptual core of our approach.

\begin{theorem}\label{thm:mainthem}
	There is a deterministic $(3+\eps)$-approximation algorithm for fair sum of radii with outliers running in time $2^{O(k \log (k/\eps))}\cdot \mathrm{poly}(n)$. 
\end{theorem}

At a high level, our algorithm uncovers the hidden structure of an optimal clustering and settles clusters one by one while maintaining feasibility with respect to fairness constraints.
We first transform the original instance into a colorful formulation that enforces the selection of exactly one center from each group.
For this structured problem, we develop a novel \emph{\fpt\ iterative ball-finding framework} that leverages \fpt\ time to exhaustively probe dense balls of appropriate radius.
A key structural lemma shows that these dense balls reveal one optimal cluster in every iteration, enabling the algorithm to recover the full solution. See~\Cref{ss:techov} for a technical overview.

Our framework extends beyond the sum of radii objective: without any modification, it yields approximation guarantees for any monotone symmetric norm of the radii, which we denote by \emph{\fairnoroutsc}. Moreover, the algorithm is \emph{oblivious} to the choice of norm.

\begin{theorem}[Main Result]\label{thm:mainthemnorm}
	There is a deterministic algorithm that computes a $(3+\eps)$-approximate solution to \fairnoroutsc in time $2^{O(k \log (k/\eps))}\cdot \mathrm{poly}(n)$. Furthermore, assuming gap-ETH, this approximation factor is tight for any \fpt\ algorithm.
\end{theorem}

First, we remark that existing algorithms for sum of radii under additional constraints are often tailored to the sum objective and do not extend to other norms, including even the max objective, which is arguably the simplest~\cite{inamdar_et_al:LIPIcs.ESA.2020.62,bandyapadhyay_et_al:LIPIcs.SoCG.2023.12,DBLP:conf/isaac/CartaDHR024,DBLP:conf/approx/BandyapadhyayC25}.

Second, we explain what we mean by our algorithm being \emph{oblivious} to the norm objective.
Since optimal solutions for different norms can be far apart, a single solution cannot, in general, approximate all norms simultaneously. We circumvent this by outputting a small list of candidate solutions that collectively cover all monotone symmetric norms. 
Specifically, given an instance of \fairnoroutsc, the algorithm outputs a family  $\cS$ of $k$ balls of size $2^{O(k \log (k/\eps))}$ in time $2^{O(k \log (k/\eps))}\cdot \mathrm{poly}(n)$ such that, for every monotone symmetric norm $f$, the list $\cS$ contains a $(3+\eps)$-approximate solution to the instance when the objective is given by $f$.
We also note that similar list-based approaches have been developed for the $k$-means objective in Euclidean space~\cite{DBLP:journals/jco/FengHHW19,bhattacharya2018faster,ding2020unified}. 
While these algorithms output a list of candidate solutions that handle multiple constraints for $k$-means, they rely crucially on the geometric structure of Euclidean space and the specific properties of the $k$-means objective.
In contrast, our framework targets a specific combination of constraints, but is essentially oblivious to both the objective and the underlying metric, and instead exploits structural properties of optimal clustering via the FPT iterative ball-finding framework.

Finally, our framework extends to a more general fairness notion, called \emph{fair-range} constraints, in which each demographic group is subject to both lower and upper bound requirements. This notion has also been recently studied for objectives such as $k$-center, $k$-median, and $k$-means~\cite{DBLP:conf/icml/HotegniMV23, DBLP:journals/corr/NNJ,
	DBLP:conf/nips/ZhangCLCHF24,DBLP:journals/corr/TGOG,DBLP:journals/corr/GT26}.

\begin{restatable}{theorem}{fairangethm}\label{thm:fairrange}
	There is a deterministic $(3+\eps)$-approximation algorithm for fair-range norm of radii with outliers running in time $2^{O(k \log (k/\eps))}\cdot \mathrm{poly}(n)$.
\end{restatable}

The proof follows by reducing fair-range constraints to unit-requirement groups and applying the same structured reduction and iterative ball-finding algorithm developed for the upper-bound setting (see Section~\ref{sec:extn}).

\paragraph*{Hardness of Approximation.}
The hardness result of~\Cref{thm:mainthemnorm} follows from the hardness of approximation for fair \kcen with outliers, which is a special case of \fairnoroutsc. Specifically, Charikar et al.~\cite{DBLP:conf/soda/CharikarKMN01} studied $k$-center with outliers in the presence of \emph{forbidden centers}, where certain points are disallowed from being chosen as centers. They proved that for any $\varepsilon>0$, it is NP-hard to approximate this problem within a factor of $3-\varepsilon$.
Moreover, their reduction preserves the parameter $k$. It originates from \textsc{Max $k$-Coverage}, which is hard to approximate to a factor better than $1-1/e$ for \fpt\ algorithms under gap-ETH~\cite{Man20}.\footnote{Informally, gap-ETH says that there exists an $\epsilon>0$ such that no $2^{o(N)}$ time algorithm for $3$-SAT on $N$ variables can distinguish whether a given $3$-SAT formula has a satisfying assignment or every assignment satisfies at most $(1-\epsilon)$ fraction of the clauses.} 

The forbidden-centers variant can be viewed as a special case of fair $k$-center with outliers. Indeed, partition the point set into two groups: one consisting of forbidden centers with an upper bound of $0$, and the other containing all remaining points with an upper bound of $k$. Any feasible fair solution is then forced to select all centers from the allowed group, exactly reproducing the forbidden-centers constraint. Since the reduction preserves the parameter $k$, it implies that no \fpt\ algorithm can achieve a $(3-\varepsilon)$-approximation for fair $k$-center with outliers under gap-ETH.

\subsection{Technical Overview}\label{ss:techov}

We present the key ideas for the sum of radii objective (i.e.,~\Cref{thm:mainthem}). 
The algorithm for the norm objective builds on the same ideas and is presented in the main technical section.

\subsubsection{Challenges in Existing Techniques}

As mentioned earlier, the sum of radii (SoR) objective has been studied separately under fairness and outlier constraints. 
A natural approach is therefore to adapt these techniques to our setting, where both constraints are present simultaneously. 
However, both directions encounter fundamental obstacles.

\paragraph*{Projection-based approaches.}
A prominent paradigm in fair clustering is the \emph{projection} framework~\cite{DBLP:conf/icml/JonesNN20,DBLP:conf/www/GadekarGT25}, which proceeds in two phases: first, compute an approximate solution ignoring fairness constraints; second, project this solution to a fair one, typically incurring an additional additive loss in the approximation factor.
For the fair sum of radii problem, \cite{ChenXXZ24} first obtain a $(2+\eps)$-approximation in \fpt\ time ignoring fairness constraints, which is then converted into a $(3+\eps)$-approximate fair solution. 

Concretely, given an approximate (non-fair) center set $T$, each center in $T$ is matched to a nearby center in an optimal fair solution, exploiting the fact that such a center must exist within the radius of its optimal ball.
However, this approach fundamentally fails in the presence of outliers.
Even though a $(3+\eps)$-approximation is known for unconstrained sum of radii with outliers~\cite{BuchemERW24}, the centers in $T$ may correspond to points that are outliers in the optimal fair solution.
In such cases, there is no nearby fair center to match to, and the projection step breaks down.

\paragraph*{Outlier-based approaches.}
Another natural direction is to adapt combinatorial algorithms for clustering with outliers.
Consider the classical algorithm of Charikar et al.~\cite{DBLP:conf/soda/CharikarKMN01} for $k$-center with outliers.
Their algorithm iteratively selects a densest ball of radius\footnote{We denote by $\opt$, the optimal cost of the problem's instance.} $\opt$, expands it to radius $3\opt$, and removes the covered points.
A key structural property is that optimal clusters can be ordered as $\pi^*_1,\dots,\pi^*_k$ such that for every $j \in [k]$,
\[
\big|\bigcup_{i=1}^j E_i\big| \;\ge\; \big|\bigcup_{i=1}^j \pi^*_i\big|.
\]

Intuitively, if the selected ball intersects the densest remaining optimal cluster $\pi^*_j$, then its expansion by a factor $3$ fully covers $\pi^*_j$, by triangle inequality.
Otherwise, the ball is disjoint from all remaining optimal clusters, and by a density argument must contain at least as many points as $\pi^*_j$, allowing these points to be charged to $\pi^*_j$.

Extending this approach to the fair setting introduces new difficulties.
A natural idea is to guess the color of the center of the densest remaining optimal cluster, which leads to an \fpt\ running time in $k$.
However, even after restricting to a fixed color, the selected densest ball may intersect a \emph{different} optimal cluster whose center has another color.\footnote{This is not a problem in analyzing the algorithm of~\cite{DBLP:conf/soda/CharikarKMN01} since the ordering of the optimal clusters is chosen after the algorithm terminates. In our case, if we chose the ordering later, then the ball's center color may mismatch with that of the intersecting optimal cluster.}
In this case, expanding the ball prematurely settles the wrong cluster, while the intended cluster remains uncovered.
Thus, the color constraints are violated and there is no mechanism to correct this choice.
Furthermore, in the sum of radii objective, unlike $k$-center, one cannot freely charge the cost of a solution ball to an optimal cluster multiple times, which introduces additional complications.

These challenges motivate a different approach. In contrast, we design an algorithm that {directly constructs a fair solution}, bypassing two-phase framework altogether. Our approach consists of two main components: an approximation-preserving reduction to a highly structured colorful instance, and a  \fpt-time iterative ball-finding algorithm for solving this structured problem.

\subsubsection{From fair constraints to colorful structure}\label{ss:techov:redcol}

Our first step is an approximation-preserving reduction from \fairsoroutsc to a structured instance with unit group requirements. 
The main goal of this step is to reduce the instance to exactly $k$ groups with unit requirements while preserving feasibility and cost.

Towards this, we first reduce the problem to an instance with unit group requirements, and then merge the groups to obtain exactly $k$ groups, all while preserving feasibility and cost.
To maintain the outlier constraint, we reduce the problem to a supplier version. 
In this version, we are given two sets $F$ and $C$, where $F$ is partitioned into $t$ demographic groups. 
The task is to select $k$ centers $S \subseteq F$ and radii such that the union of the corresponding balls covers at least $|C|-z$ points from $C$, while satisfying the group constraints $|S \cap F_i| \le 1$ for all $i$.

The reduction proceeds by duplicating each demographic group $X_i$ according to its allowed quota, creating multiple copies in $F$, and producing a collection of groups from which at most one center can be selected per copy.
We set $C = X$. 
Any feasible solution to this supplier instance with unit group requirements is also feasible for the original problem, and vice versa.
A major difficulty is that this duplication can create a number of groups far exceeding $k$, making the instance unsuitable for algorithms whose complexity depends exponentially on the number of groups.
In particular, it is not clear whether one can further reduce this problem to exactly $k$ groups while preserving the optimal cost using only polynomial-time transformations.

\paragraph*{Merging groups in \fpt\ time.}
Here, the power of \fpt\ time in $k$ becomes crucial.
We apply the color-coding technique to randomly assign each group in $F$ to one of $k$ colors, and merge all groups assigned the same color into a single color class with unit requirement.
This produces an instance of \emph{\colkcenosc}, where a feasible solution must select exactly one center from each color class.
Standard arguments show that, with probability at least $2^{-\Omega(k)}$, this transformation preserves the optimal cost. 
Moreover, with the same probability, any feasible solution to the original instance remains feasible in the resulting colorful instance.

The reduction to unit group requirements is essential for this step.
Since each group contributes at most one center, merging groups preserves feasibility by enforcing the selection of exactly one center per color class.
In contrast, applying color coding directly to groups with general upper bounds would require determining appropriate bounds for the merged groups, which depend on the unknown optimal solution.

By derandomizing this procedure, we generate $2^{O(k)}$ colorful instances, one of which preserves the optimal cost of the original instance. 
Thus, \fpt\ time allows us to compress the instance to exactly $k$ groups without losing optimality.
Consequently, it suffices to design an \fpt\ approximation algorithm for the colorful problem.



\subsubsection{$3$-approximation for \colkcenosc in FPT time}\label{ss:techov:3apxcol}

For a cleaner exposition, we work with the supplier version, where $X = (F \cup C)$, and the solution consists of $k$ balls centered at facilities in $F$ that cover at least $n-z$ points from $C$, where $n = |C|$. 
Moreover, $F$ is partitioned into $k$ color classes $F_1,\dots,F_k$, and the solution must select exactly one center from each $F_i$.
Let $(O^*, R^*)$ be an optimal solution, and let $\Pi^* = \{\pi^*_1,\dots,\pi^*_k\}$ be the corresponding optimal clustering. 
Without loss of generality, assume that the center of $\pi^*_i$ has color $i$ and radius $r^*_i$. 
Using standard discretization, we assume that approximate values of $R^*$ are known.
At a high level, the algorithm (see~\Cref{algo:colksupalgo}) proceeds in phases, settling at least one optimal cluster in each phase. 
In a phase corresponding to cluster $\pi^*_i$, the algorithm adds a ball of appropriate radius that fully accounts for $\pi^*_i$, and removes its covered points. 
Thus, after at most $k$ phases, all clusters are settled.

\paragraph*{Coverage via Charging.}
To bound coverage, we use a charging scheme similar to~\cite{DBLP:conf/soda/CharikarKMN01}. 
Each point in an optimal cluster either charges itself (if it is covered by a chosen ball) or is charged to a \emph{free point}, i.e., a point that is either an outlier or does not belong to any remaining optimal cluster.
A key property of this scheme is that each point is charged at most once: once a point is used for charging, it is removed from further consideration. 
Thus, no free point is used more than once, and points belonging to optimal clusters are also charged at most once. 


\paragraph*{Our Algorithm.}
The main technical contribution is an \emph{iterative ball-finding framework} that uncovers the structure of the optimal clustering in \fpt\ time.
In each phase, we implicitly target the densest remaining optimal cluster $\pi^*_j$. 
While $\pi^*_j$ is unknown, we can guess its color $j \in [k]$, allowing us to restrict attention to centers in $F_j$.
Next, the algorithm invokes the following iterative ball-finding subroutine to construct a set of disjoint balls of radius $r^*_j$ centered at points in $F_j$.

\paragraph*{Iterative Ball-Finding via FPT Search.}
Figure~\ref{fig:overview} illustrates the high-level structure of our iterative ball-finding subroutine.
In more details, for the correct guess $j \in [k]$ of the color the optimal $\pi^*_j$, we exploit the fact that the densest ball of radius $r^*_j$ in the remaining points with center color  $j$ must contain as many points as  the remaining points of $\pi^*_j$.
Accordingly, the subroutine repeatedly constructs a set $\cB$ of (up to) $4k$ disjoint balls of radius $r^*_j$ centered at (disjoint) points in $F_j$, always choosing the densest  ball in the remaining points.\footnote{For e.g., by temporarily deleting the points and the center of the ball.}

\emph{Terminology.}
Let $B \in \cB$. 
We say:
\begin{itemize}
	\item $B$ is \emph{nearby} to $\pi^*_j$ if its center is within distance $2r^*_j$ of the optimal center $o^*_j$;
	\item $B$ is \emph{good} if it contains at least as many free points as the remaining points of $\pi^*_j$;
	\item $B$ is \emph{dense} if more than half of its points belong to remaining optimal clusters, and \emph{light} otherwise.
\end{itemize}

\paragraph*{Structural Trichotomy.}
We prove that among the balls in $\cB$, at least one of the following configurations must occur:
\begin{itemize}
	\item a {nearby ball} to $\pi^*_j$;
	\item a {good ball} for $\pi^*_j$;
	\item  two {light balls} for $\pi^*_j$ whose union contains sufficient {free} points to simulate $\pi^*_j$ and both of which are nearby balls to (another) uncovered optimal cluster $\pi^*_{\ell}$.
\end{itemize}

\begin{figure}[h!]
	\centering
	\begin{tikzpicture}[
		node distance=2.2cm,
		every node/.style={
			draw, rectangle, rounded corners,
			align=center, minimum width=3.0cm, minimum height=0.9cm
		},
		case/.style={fill=gray!15},
		arrow/.style={->, thick}
		]
		
		\node (start) {Remaining points};
		\node (pick) [below of=start] {Implicitly target densest\\uncovered optimal cluster $\pi^*_j$};
		\node (balls) [below of=pick] {Construct densest disjoint balls\\of radius $r^*_j$ of  center color $j$};
		
		\node[case] (near) [below left=2.2cm and 1.8cm of balls] {Nearby ball exists};
		\node[case] (good) [below=2.2cm of balls] {Good ball exists};
		\node[case] (two)  [below right=2.2cm and 1.8cm of balls] {Two light balls close\\to same cluster $\pi^*_\ell$};
		
		\node (settle1) [below of=near] {Expand to $3r^*_j$\\settle cluster $\pi^*_j$};
		\node (settle2) [below of=good] {Expand to $3r^*_j$\\settle cluster $\pi^*_j$};
		\node (replace) [below of=two] {Find correct-colored\\nearby center for $\pi^*_\ell$\\Expand to $3r^*_j+3r^*_\ell$\\settle cluster $\pi^*_\ell$};
		\node (settle3) [below of=replace] {Zero radius cluster at \\center of any light ball\\settle cluster $\pi^*_j$};
		
		\node (remove) [below=2.4cm of settle2] {Remove covered points and centers\\Iterate on remainder};
		
		\draw[arrow] (start) -- (pick);
		\draw[arrow] (pick) -- (balls);
		
		\draw[arrow] (balls) -- (near);
		\draw[arrow] (balls) -- (good);
		\draw[arrow] (balls) -- (two);
		
		\draw[arrow] (near) -- (settle1);
		\draw[arrow] (good) -- (settle2);
		\draw[arrow] (two) -- (replace);
		\draw[arrow] (replace) -- (settle3);
		
		\draw[arrow] (settle1) |- (remove);
		\draw[arrow] (settle2) -- (remove);
		\draw[arrow] (settle3) |- (remove);
		
	\end{tikzpicture}
	\caption{Flow of the iterative ball-finding algorithm for \colkcenosc. Lightly shaded boxes correspond to the three cases of the structural trichotomy, each enabling the algorithm to place a factor $3$ radius ball that settles at least one optimal cluster while preserving color constraints.}
	\label{fig:overview}
\end{figure}

Figure~\ref{fig:ballstri} shows an execution of the iterative ball-finding subroutine.
The existence of this trichotomy follows from a counting argument. 
If fewer than $4k$ disjoint balls can be constructed, then one of them must be centered at the optimal center of $\pi^*_j$, yielding a nearby ball. 
Otherwise, consider the case where $|\cB| = 4k$ and no nearby ball exists. 
Since each ball is chosen to be densest, each contains at least as many points as the remaining points of $\pi^*_j$. 
If many balls are dense, then their total contribution from remaining optimal clusters exceeds the total number of such points, leading to a contradiction. 
Thus, at least $2k$ balls must be light. 
If one of them is good, we obtain the second case. 
Otherwise, by a pigeonhole argument, two light balls must intersect the same cluster $\pi^*_\ell$, yielding the third case. 
Moreover, since there is no nearby ball, all balls in $\cB$ are disjoint from $\pi^*_j$, and hence each contains at least as many points as the remaining points of $\pi^*_j$. 
Since each light ball has at least half of its points as free points, the union of the two light balls contains at least as many free points as the remaining points of $\pi^*_j$.

		
		
		
		
		
		
		
		

\begin{figure}
	\centering
	\includegraphics[width=0.8\linewidth]{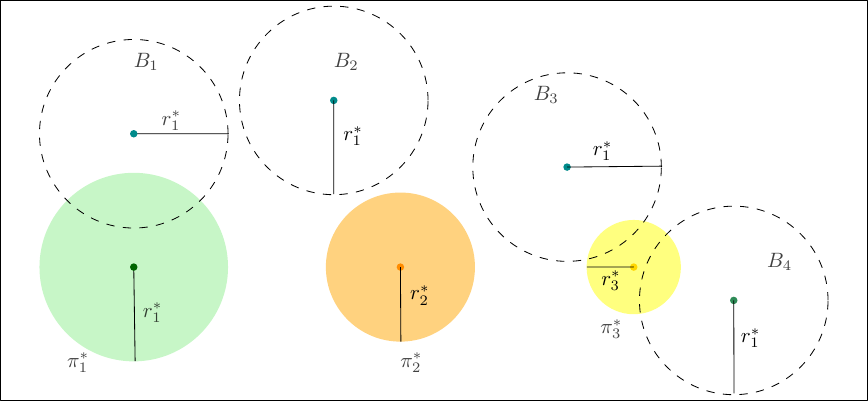}
	\caption{\small The rectangle represents the full point set, while the colored discs correspond to the optimal clusters, with the remaining area representing outliers.
		The relative sizes of the discs illustrate the densities of the clusters. For simplicity, the center color classes are not shown.
		The iterative ball-finding subroutine is invoked for the densest uncovered optimal cluster, which is  green ($\pi^*_1$), and some balls constructed by the subroutine are depicted by dashed boundaries.
		Here, $B_1$ is a nearby ball to $\pi^*_1$, $B_2$ is a good dense ball for $\pi^*_1$, and $B_3$ and $B_4$ are light balls intersecting another uncovered cluster $\pi^*_3$, whose centers both lie within distance $(r^*_1+r^*_3)$ of the center of $\pi^*_3$.
	}
	\label{fig:ballstri}
\end{figure}

\paragraph*{Using the trichotomy.}
Each case of the trichotomy yields a way to construct a feasible solution while making progress.

\begin{itemize}
	\item \textbf{Nearby ball.}
	Let $\hat{B}$ be a nearby ball to $\pi^*_j$, and let $\hat{t}$ be its center. 
	Since $\hat{t} \in F_j$, selecting $\hat{t}$ preserves the color constraint for cluster $\pi^*_j$. 
	By expanding $\hat{B}$ to radius $3r^*_j$, we obtain a ball that fully covers $\pi^*_j$. 
    	We charge points of $\pi^*_j$ to themselves.
	
	\item \textbf{Good ball.}
	Let $\hat{B}$ be a good ball for $\pi^*_j$ with center $\hat{t} \in F_j$. 
	Although $\hat{B}$ may not intersect $\pi^*_j$, it contains enough free points to simulate it. 
	We expand $\hat{B}$ to radius $3r^*_j$. 
	We charge the points of $\pi^*_j$ as follows: points already covered by previously chosen balls charge themselves, while the remaining points are charged to distinct free points in the expanded ball. 
	Such a charging is possible since $\hat{B}$ contains at least as many free points as the remaining points of $\pi^*_j$. 
	Since $\hat{t}$ has color $j$, the color constraint is preserved.
	
	\item \textbf{Two light balls.}
	Let $\hat{B}_1$ and $\hat{B}_2$ be two light balls that intersect another cluster $\pi^*_\ell$. 
	Let $\hat{t}_1$ be the center of $\hat{B}_1$, where $\hat{t}_1 \in F_j$. 
	We construct a ball $\hat{E}$ centered at $\hat{t}_1$ with radius $3r^*_j + 2r^*_\ell$, which covers both $\hat{B}_1 \cup \hat{B}_2$ and the cluster $\pi^*_\ell$.
	
	We charge the points of $\pi^*_j$ as follows: points already covered charge themselves, while the remaining points are charged to free points in $\hat{B}_1 \cup \hat{B}_2$. 
	This is feasible since the two light balls together contain enough free points by the structural trichotomy. 
	The points of $\pi^*_\ell$ are fully covered by $\hat{E}$ and charge themselves. 
	
	To satisfy the color constraint for $\pi^*_\ell$, we additionally open a zero-radius ball at any center in $F_\ell$. 
	Thus, we settle both clusters $\pi^*_j$ and $\pi^*_\ell$ while preserving the color constraints.
	
\end{itemize}

In all cases, we settle at least one optimal cluster while maintaining feasibility. 
In the third case, we may settle two clusters simultaneously, yielding a total cost bounded by $3r^*_j + 2r^*_\ell$. 
Since each optimal radius is charged at most once across all iterations, this implies a $3$-approximation.
For general norm objectives, this step requires additional care (see~\Cref{ss:analysis}).

\subsubsection{The final algorithm}\label{ss:techov:final}

Using the reduction from~\Cref{ss:techov:redcol}, we first transform the given instance of fair SoR with outliers into $2^{O(k)}$ instances of colorful SoR with outliers in deterministic time $2^{O(k)}\poly(n)$. 
To solve each colorful instance, we combine two ingredients: (i) enumeration over approximate radii, and (ii) branching over the choices in the iterative ball-finding procedure.

For (i), it is known that, in time $2^{O(k\log(k/\eps))}\poly(n)$, one can construct a list $\cL$ of $k$-tuples of radii that contains an approximate optimal radii vector~\cite{DBLP:journals/corr/FG25}. 
For (ii), the branching arising from the iterative ball-finding framework can be bounded by $2^{O(k\log(k/\eps))}$.

We therefore proceed as follows: for each colorful instance and for each candidate radii vector in $\cL$, we run the branching version of the algorithm described in~\Cref{ss:techov:3apxcol}. 
Among all solutions obtained, we return a feasible solution that minimizes the given norm of the radii.
The overall running time is bounded by $2^{O(k \log (k/\eps))}\cdot \mathrm{poly}(n)$.

We remark that the above procedure naturally yields a small list of candidate solutions that simultaneously covers all monotone symmetric norms. 
Specifically, for each candidate radii vector in $\cL$, we store the corresponding feasible solution (if any) produced by the algorithm across all colorful instances. 
For any fixed norm $f$, the optimal solution corresponds to some radii vector in $\cL$, and hence the corresponding stored solution provides a $(3+\eps)$-approximation for the objective defined by $f$. 
\section{Preliminaries}

Given a metric space $(X,d)$, for $x \in X$ and $r \in \mathbb{R}_{\geq 0}$, we denote by $\ball(x,r)$ the ball centered at $x$ with radius $r$, i.e., $\ball(x,r):= \{p\in X: d(x,p)\le r\}$. Additionally, for $Y \subseteq X$, let $\ball_Y(x,r) = \ball(x,r) \cap Y$.
We now define the problems of our interest.

\begin{definition}[\fairsorout]
	Given a metric space \met on $n$ points, two non-negative integers $k$ and $\outcnt$, and a partition $\groups=\{X_1,\dots, X_t\}$ of \points along with an integer $k_i \geq 0$ for each group $X_i \in \groups$, the task is to find a set $S=\{s_1,\dots,s_k\} \subseteq \points$ of size $k$ and a radii set $R=\{r_1,\dots,r_k\} \subseteq \mathbb{R}_{\ge 0}$ so as to minimize $\sum_{i \in [k]} r_i$, such that:
	\begin{itemize}
		\item (fairness constraints) $|S \cap X_i| \leq k_i$ for every $X_i \in \groups$, and
		\item (outlier constraint) $\big|\bigcup_{i \in [k]} \ball_X(s_i,r_i)\big| \ge n-z$.
	\end{itemize}
\end{definition}

We denote by $\cI=\foins$ an instance of \fairsoroutsc.  
For a given solution $(S,R)$, we define its cost as $\cost_{\cI}((S,R)) = \sum_{r \in R} r$.  
The points covered by $(S,R)$, i.e., those in $\bigcup_{i \in [k]}\ball_X(s_i,r_i)$, are called \emph{inliers}, and the remaining points are called \emph{outliers}, denoted by $Z$.
Furthermore, $(S,R)$ is said to be a \emph{feasible solution} if $|S| = k$ and it satisfies both the fairness constraints and the outlier constraint, i.e., $|S \cap X_i| \le k_i$ for every $X_i \in \groups$ and at least $n-z$ points are covered.
Without loss of generality, we assume that $\sum_{i \in [t]} k_i \geq k$, since otherwise every set of size $k$ is infeasible.\footnote{In this case, one can decrease $k$ so that the instance becomes feasible.}

\medskip

Next we define a variant problem with unit group requirements that will be useful for our algorithms. We call this problem \colksupo, and is defined in the \emph{supplier} setting. Here the point set is partitioned into a \emph{facility set} $F$ and a \emph{client set} $C$, along with a metric $d'$ defined between $F$ and $C$. The groups are defined over $F$, and the task is to cover at least $|C|-z$ clients from $C$ using a solution $(S,R)$ with $S \subseteq F$, $|S|=k$. Moreover, there are precisely $k$ groups over $F$, and a feasible solution must select exactly one center from each group.

\begin{definition}[\colksupo]
	Given a metric space \mets, two non-negative integers $k$ and $\outcnt$, and a $k$-partition $\cF=\{F_1,\dots, F_k\}$ of $F$, the task is to find a set $S=\{s_1,\dots,s_k\} \subseteq F$ of size $k$ and a radii set $R=\{r_1,\dots,r_k\} \subseteq \mathbb{R}_{\ge 0}$ such that:
	\begin{itemize}
		\item $|S \cap F_i| = 1$ for every $F_i \in \cF$, and
		\item $\big|C \cap \bigcup_{i \in [k]} \ball(s_i,r_i)\big| \ge |C|-z$,
	\end{itemize}
	so as to minimize $\sum_{i \in [k]} r_i$.
\end{definition}

We denote by $\coins$ an instance of \colkcenosc.  
Given an instance $\cI$, we let $\opt(\cI)$ denote the optimal cost. When $\cI$ is clear, we simply write $\opt$.  
We also use $n$ to denote $|X|$ or $|F \cup C|$, depending on the context.  
When \fairsorout is defined in the supplier setting, we refer to it as \emph{\fairksupo}.

\medskip

We now generalize the objective to norm minimization. Given a monotone symmetric norm $||\cdot|| : \mathbb{R}^k \rightarrow \mathbb{R}_{\ge 0}$, the goal is to minimize $||R||$. For simplicity, we continue to use the same problem names, even though the objective is now a norm instead of the sum of radii.

\begin{remark}
	We define the problems assuming that $\met$ is a metric space. In some parts of our algorithm, we work with pseudometrics, where distinct points may have zero distance. Such situations arise due to duplication of points, and our algorithms are robust to this.
\end{remark}

\medskip

Next, we define the notion of an $\eps$-approximation of the optimal radii.  
Let $||\cdot|| : \mathbb{R}^k \rightarrow \mathbb{R}_{\ge 0}$ be a monotone symmetric norm.  
For a multiset $R^*=\{r^*_1,\dots,r^*_k\}$ of positive reals, we say that a multiset $R=\{r_1,\dots,r_k\}$ is an \emph{$\eps$-approximation} of $R^*$ if:
\begin{itemize}
	\item $r_i \ge r^*_i$ for all $i \in [k]$, and
	\item $||(r^*_1,\dots,r^*_k)|| \le ||(r_1,\dots,r_k)|| \le (1+\eps)||(r^*_1,\dots,r^*_k)||$.
\end{itemize}

It is known that if the largest entry in $R^*$ is known, then one can compute such an approximation in \fpt\ time.

\begin{theorem}[\cite{DBLP:journals/corr/FG25}]\label{thm:epsapxrad}
	There is an algorithm that, for any $\eps>0$ and a (partially unknown) $k$-multiset $R^*$ of positive reals whose largest entry is known, computes a list $\cL$ of size $2^{O(k \log (k/\eps))}$ containing a $k$-multiset that is an $\eps$-approximation of $R^*$. Moreover, the algorithm runs in time $2^{O(k \log (k/\eps))}$.
\end{theorem}

\begin{figure}[h!]
    \centering
    
    \includegraphics[width=0.7\textwidth]{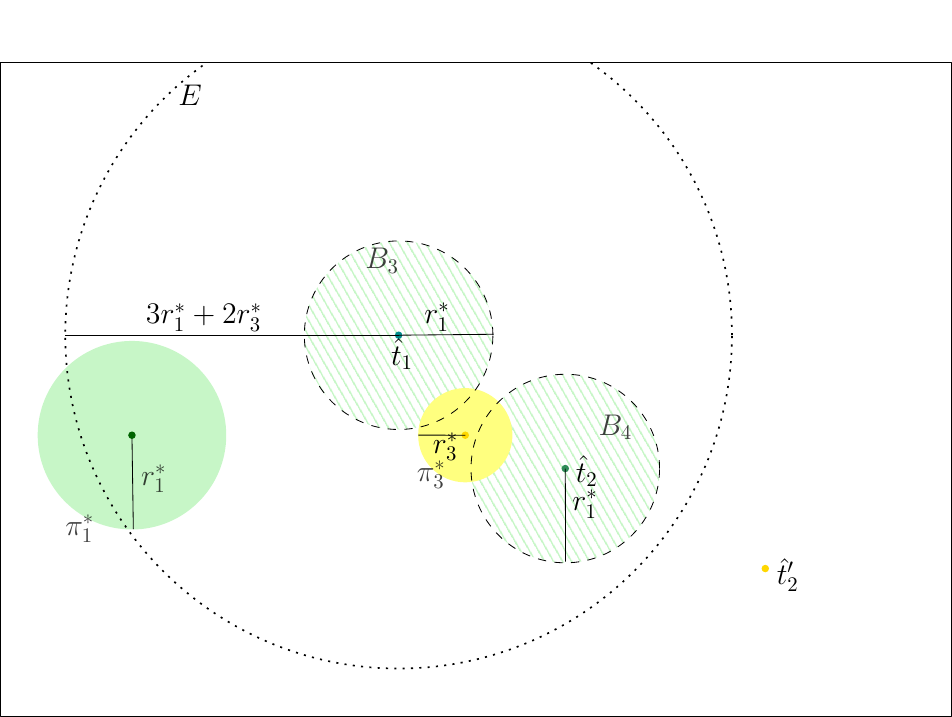}
    
    \vspace{0.5em} 
    
    \includegraphics[width=0.7\textwidth]{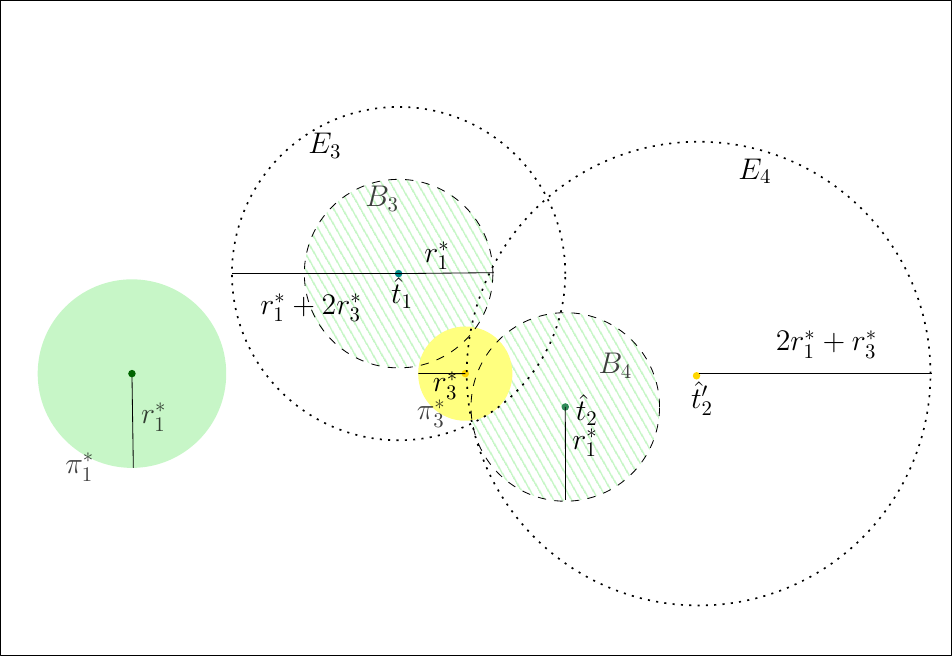}
    
    \caption{\small Illustration of the two-light-balls case for the example of~\Cref{fig:ballstri}. For clarity, we show only the relevant data.
Free points in $B_3$ and $B_4$ are indicated by striped areas.
\textbf{Top:} We open two balls: one centered at $\hat{t}_1$ with radius $3r^*_1+2r^*_3$, denoted by $E$, which contains enough points to charge all points in $\pi^*_1 \cup \pi^*_3$, and a zero-radius ball at some $\hat{t}'_2$ of yellow color to satisfy the color constraints.
\textbf{Bottom:} For monotone symmetric norms, we open balls $E_3$ and $E_4$ with radii $r^*_1+2r^*_3$ and $2r^*_1+r^*_3$, respectively. Ball $E_3$ is centered at $\hat{t}_1$, while the center of $E_4$ lies within distance $r^*_1+r^*_3$ from $\hat{t}_2$. Together, these balls contain enough points and satisfy the color constraints.}
    
    \label{fig:twolightballs}
\end{figure}

\paragraph*{Outline.}
To prove~\Cref{thm:mainthemnorm}, we first present an approximation-preserving \fpt-reduction from \fairsoroutsc\ to \colksuposc\ in~\Cref{sec:fairtocol}.  
Next, in~\Cref{sec:mainsec}, we develop a $(3+\eps)$-approximation algorithm for \colksuposc\ running in \fpt\ time parameterized by $k$, completing the proof of~\Cref{thm:mainthemnorm}.  
Finally, in~\Cref{sec:extn}, we show how our approach extends to the more general fair-range setting.

\section{From \fairsorout to \colksupo}\label{sec:fairtocol}

In this section, we  show the following  approximation preserving reduction  for \fpt\ algorithms from \fairsoroutsc to \colksuposc. 
\begin{theorem}\label{thm:redfc}
	There exists a randomized polynomial-time algorithm that maps an instance $\cI=\foins$  of \fairsoroutsc to an instance  $\cJ = \coins$ of \colksuposc with the following property: with probability at least $\frac{1}{2^{\Theta(k)}}$ over the randomness of the reduction, every $\alpha$-approximate solution for $\cJ$ can be converted in polynomial time into an $\alpha$-approximate solution for $\cI$, for any $\alpha\geq 1$. Furthermore, the reduction can be derandomized in time $2^{O(k)}\poly(|X|)$.
\end{theorem}

	
	

We prove~\Cref{thm:redfc} in three steps. First, in~\Cref{ss:redfu}, we show a polynomial time reduction from \fairsoroutsc to the supplier version, \fairksuposc, that has all group upper bounds equal to one, i.e., $k_i = 1$, while preserving the feasible solutions and their costs. 
Then, in~\Cref{ss:reduc}, we show a reduction from this problem to \colksuposc that preserves the optimal cost. Finally, we combine these two reductions to finish the proof of~\Cref{thm:redfc}, which is deferred to~\Cref{ss:redpf}.

\subsection{Reducing \fairsoroutsc to unit group supplier}\label{ss:redfu}
We show the following reduction.
\begin{lemma}\label{lem:redfu}
	There exists a polynomial time reduction that given an instance $\cI$ of \fairsoroutsc, constructs an instance $\cL$ of \fairksuposc with unit group requirement such that for any feasible solution $(S,R)$ to $\cI$, there exists a set $(T,R)$ of facilities in $\cL$ such that $(T,R)$ is a feasible solution to $\cL$ and $\cost_{\cL}((T,R)) = \cost_{\cI}((S,R))$, and vice versa. In particular, the optimal cost of $\cL$ is equal to that of $\cI$.
\end{lemma}
\begin{proof}
	The reduction proceeds simply by duplicating each group $X_i$ in $\cI$ $k_i$ many times and creating appropriate copies of each group. In more detail, let $\cI=\foins$ be the given instance of \fairsoroutsc. We construct an instance $\cL=\fosupins$ of \fairksuposc with unit group requirements as follows. 
	
	\noindent\textbf{Facility set $F$.}
	For every group $X_i \in \cX$ and every point $x \in X_i$, create $k_i$ distinct copies 
	\[
	(x,i,1),\dots, (x,i,k_i)
	\]
	and let $F= \{(x,i,j)\mid x \in X_i, j \in [k_i]\}$.
	
	\noindent\textbf{Facility groups.}
	For each $i \in [t]$ and $j \in k_i$, define group $F_{i,j} := \{(x,i,j) \mid x \in X_i\}$ with requirement $1$. Note that $\{F_{i,j}\}$ is a partition of $F$.
	
	\noindent\textbf{Client set $C$.} Set $C:=X$
	
	\noindent\textbf{Distance $d'$.} Naturally extend $d$ to obtain $d'$ over $F \cup C$. Formally, for $u,v \in C \cup F$,
	\[
	d'(u,v) = \begin{cases}
		d(u,v) & \text{if } u,v \in X\\
		d(u,x) & \text{if } u \in C, v=(x,i,j) \in F\\
		d(x,y) & \text{if } u=(x,i,j) \in F, v=(y,i',j') \in F \text{ and } x \neq y.
	\end{cases}
	\]

	Note that $|F| = \sum_{i \in [t]} |X_i|\cdot k_i \le |X|^2$, since $k_i \le |X|$, for $i \in [t]$, and $\ell = \sum_{i \in [t]} k_i \le |X|t$. Thus $|F \cup C| = O(|X|^2)$. 
	
	\textbf{Mapping between the feasible solutions.}
	Next we define mappings to obtain a feasible solution $(T,R)$ to $\cL$ given a feasible solution $(S,R)$ to $\cI$, and vice versa. First consider a feasible solution $(S,R)$ to $\cI$. For each $X_i \in \cX$ and for each $x \in S \cap X_i$, choose a distinct copy $(x,i,j)$ from $F_{i,j}$ and add it to $T$. Since $(S,R)$
	is feasible to $\cI$, we have that $(T,R)$ is feasible to $\cL$.  Conversely, consider a feasible solution $(T,R)$ to $\cL$.  For each selected facility copy $(x,i,j) \in T$, choose its original point $x \in X$ and add it to $S$ (removing duplicates). Since each group $F_{i,j}$ contributes to at most one facility in $T$, we have that $|S \cap X_i| \le k_i$, since there are $k_i$ many groups $F_{i,j}$ corresponding to $X_i$. Therefore, $(S,R)$ is a feasible solution to $\cL$.
	For a feasible solution $(S,R)$ to $\cI$, we say the feasible solution $(T,R)$ to $\cL$ obtained using this mapping as the \emph{image of $(S,R)$ in $\cL$}, and vice versa (by breaking ties arbitrary among the duplicate copies).
	
	Next, consider   solutions $(S,R)$ to $\cI$ and  $(T,R)$ in $\cL$ such that they are images of each other, and let $(x,i,j) \in T$ and $x\in S$.
	We have a simple observation that for any $p \in C$, we have  $d'(p,(x,i,j)) = d(p,x)$. Therefore, $\ball_C((x,i,j),r)=\ball_X((x,i,j),r)=\ball_X(x,r)$. Hence, the number of points covered by $(S,R)$ in $X$ is same as that covered by $(T,R)$ in $C$, i.e., $\cup_{((x,i,j),r) \in (T,R)}  \ball_C((x,i,j),r) =\cup_{(x,r)\in (S,R)} \ball_X(x,r)$. Thus, $(S,R)$ is a feasible solution to $\cI$ if and only if $(T,R)$ is a feasible solution to $\cL$. Finally, note that $\cost_\cI((S,R)) = \cost_{\cL}((T,R))$, as desired.

\end{proof}

\subsection{Reducing unit group supplier to \colksuposc}\label{ss:reduc}
Next we reduce the number of groups in the supplier setting to $k$.

\begin{lemma}\label{lem:reduc}
	There exists a randomized polynomial time reduction that given an instance $\cL=\fosupins$ of \fairksuposc and with unit group requirements, constructs an instance $\cJ=\coinsp$ of \colkcenosc such that $(i)$ a feasible solution to $\cJ$ is a feasible solution to $\cL$, and $(ii)$    with probability $\frac{1}{2^{\Theta(k)}}$ the optimal cost of $\cJ$ is same as that of $\cL$. Furthermore, the reduction can be derandomized in time $2^{O(k)}\poly(|\cL|)$.
\end{lemma}
\begin{proof}
	The reduction proceeds by color-coding the groups of $\cL$ using $k$ colors. More specifically, $\cJ$ is exactly the same as $\cL$, except the groups. The group set $\cF'$ in $\cL$ are constructed as follows:    assign each group in $\cF$ a uniformly random color from $[k]$. Then, let $F'_i \in \cF'$ be the union of groups of $\cF$ that receive color $i \in [k]$. This completes the construction of $\cJ$.
	
	First note that any feasible solution to $\cJ$ is also a feasible solution to $\cL$ since each group in $\cJ$ is a collection of groups of $\cL$ with unit group requirement, thereby implying $(i)$. This further implies that $\opt(\cL) \le \opt(\cJ)$, finishing one direction of $(ii)$. For the other direction, fix an optimal solution $(O_L,R^*)$  to $\cL$, and let $\hat{\cF} \subseteq \cF$ denote the set of groups in $\cL$ that have an element in $O_L$. We say $\hat{\cF}$ is \emph{colorful} in $\cJ$ if every element in $\hat{\cF}$ received a distinct color in the above mentioned random process. Then,
	\[
	\Pr[\hat{\cF} \text{ is colorful in } \cJ] =\frac{k!}{k^k} = 2^{-\Theta(k)}.
	\]
	Now conditioning on the event that $\hat{\cF}$ is colorful in $\cJ$, we have that $(O_L,R^*)$ is a feasible solution to $\cJ$. Moreover, in this case,
	\[
	\opt(\cJ) \le \cost_{\cJ}((O_L,R^*)) = \cost_{\cL}((O_L,R^*)) = \opt(\cL),
	\]
	where the first equality follows since the metric in $\cJ$ is the same as in $\cL$. Therefore, with probability $1/2^{\Theta(k)}$,    it holds that $\opt(\cJ)=\opt(\cL)$, finishing the proof.    
\end{proof}

\subsection{Proof of~\Cref{thm:redfc}}\label{ss:redpf}
The proof follows by chaining the reduction of~\Cref{lem:redfu} with that of~\Cref{lem:reduc}. Given an instance $\cI=\foins$ of \fairsoroutsc, let $\cL=\fosupins$ be an instance of \fairksuposc obtained from~\Cref{lem:redfu} on $\cI$.
Furthermore, let $\cJ=\coinsp$ be the instance of \colksuposc obtained from~\Cref{lem:reduc} on $\cL$ conditioned on the event that item (ii) of~\Cref{lem:reduc} holds. Consider an $\alpha$-approximate solution $T$ to $\cJ$. Then, we have that $T$ is a feasible solution to $\cL$ from~\Cref{lem:reduc}.  We have,
\begin{align}
	\cost_{\cJ}(T) &\le \alpha \cdot \opt(\cJ)\nonumber\\
	&= \alpha \cdot\opt(\cL) \qquad \text{by item $(ii)$ of~\Cref{lem:reduc}}\nonumber\\
	&= \alpha \cdot \opt(\cI) \qquad\text{by~\Cref{lem:redfu}} \label{eq:costoft}
\end{align}
Now consider the image $S$ of $T$ in the instance $\cI$. From~\Cref{lem:redfu}, it follows that $S$ is a feasible solution to $\cI$. Moreover,
\begin{align*}
	\cost_{\cI}(S) = \cost_{\cL}(T)  = \cost_{\cJ}(T) \ \le \alpha \cdot \opt(\cI),
\end{align*}
where the first equality follows from~\Cref{lem:redfu}, the second equality follows from~\Cref{lem:reduc}, and the inequality follows from~\Cref{eq:costoft}. 

Finally, note that  the color-coding procedure can be derandomized in time $2^{O(k)} |X|\log |X|$ by using standard tools from parameterized complexity~\cite{CyganFKLMPPS15}.
\qed

\section{An FPT $3$-approximation for \colkceno}\label{sec:mainsec}

In this section, we formalize the iterative ball-finding framework outlined in~\Cref{ss:techov} and present our algorithm for the colorful SoR with outliers problem.
Recall that the colorful formulation enforces the selection of exactly one center from each color class.

\subsection{The Algorithm}
The pseudocode for our main algorithm is described in~\Cref{algo:colksupalgo}. Initially, all the optimal clusters are unmarked. Then, the algorithm calls the iterative ball-finding subroutine, described in~\Cref{algo:settle cluster}, for every unmarked optimal cluster.
As described in~\Cref{ss:techov}, this subroutine builds a collection of dense balls of the corresponding optimal cluster radius over the remaining points and examines their geometric and combinatorial structure.
The structural trichotomy guarantees that among these balls, one can either directly recover an optimal cluster or identify a small configuration whose expansion settles at least one cluster up to a constant-factor loss.

\begin{algorithm}[h]
	\caption{A $3$ \fpt\ approximation for \colksuposc} \label{algo:colksupalgo}
	\KwData{Instance $\cI=\coins$ of \colksuposc}
	\KwResult{A $3$ approximate clustering to $\cI$}
	
	$C' \leftarrow C$; 
	$\cE \leftarrow \emptyset$\;
	\texttt{// Let $\Pi^* = \{\pi_1,\dots,\pi_k\}$ be an optimal clustering to $\cI$\;}
	Unmark all the clusters in $\Pi^*$\;
	\While{$\exists$ unmarked cluster in $\Pi^*$}
	{
		$j \in [k] \leftarrow $ \guess the index of the densest unmarked optimal cluster in $C'$\label{algo:lin:den}\;
		$\cE \leftarrow \cE \cup \SC(j,F,C')$\;
	}
	\Return $\cE$\ if it covers $n-z$ points,  \textbf{else fail};
	
		
		
\end{algorithm}

\begin{algorithm}
	\caption{\SC$(j,F,C')$} \label{algo:settle cluster}
	$F'' \gets F, F''_j\gets F_j, C'' \gets C'$\;
	$\cB \gets \emptyset$\;
	\guess $r_j \gets \epsilon$-approximation of $r^*_j$\;
	\While{$\exists$ a facility in $F''_j$}
	{
		pick the densest ball $B:=\ball_{C''}(\hat{t},r_j)$ such that $\hat{t} \in F''_j$\;
		$\cB\gets \cB \cup \{B\}$\;
		$C'' \gets C''\setminus B$\;
		$F'' \gets F'' \setminus \hat{t}$\;
		\lIf{$|\cB|=4k$}{break;\qquad \texttt{// enough balls for the structural trichotomy}}
	}
	
	\guess if $\cB$ contains a ball $\hat{B}$ centered at $\hat{t} \in F_j$ that is either a \emph{nearby ball} for $\pi^*_j$ or is a \emph{good ball} for $\pi^*_j$ in $C'$\label{algo:sc:line:goodcase}\;
	\If{the \guess is successful}{
		mark $\pi^*_j$\label{algo:sc:lin:markj}\;
		$C' \gets C' \setminus \ball_{C'}(\hat{t},3r_j)$\;
		\Return $(\hat{t},3r_j)$\;}
	\Else{

		\guess if $\exists \hat{B}_1 \neq \hat{B}_2 \in \cB$ centered at $\hat{t}_1, \hat{t}_2 \in F_j$, resp. and an index $\ell \in [k]$ of unmarked optimal cluster such that both $\hat{B}_1$ and $\hat{B}_2$ are light balls for the current iteration, and are also nearby balls for $\pi^*_\ell$\label{algo:sc:lin:badcase}\;
		\lIf{the \guess is unsuccessful}{\textbf{fail}}
		\Else{
			\guess $r_\ell \gets \epsilon$-approximation of $r^*_\ell$\;        
			let $\hat{t}'_2 \in F_\ell$ be such that $d'(\hat{t}'_2,\hat{t}_2) \leq r_j+r_\ell$\tcp*{guaranteed by~\Cref{lem:charging lemma}}
			mark $\pi^*_j,\pi^*_\ell$\label{algo:lin:markell}\;
			$C' \gets C' \setminus (\ball_{C'}(\hat{t}_1,r_j+2r_\ell) \cup \ball_{C'}(\hat{t}'_2,2r_j+r_\ell))$\;
			\Return $(\hat{t}_1,r_j+2r_\ell) \cup (\hat{t}'_2,2r_j+r_\ell)$
		}
	}
	
\end{algorithm}

Let $\cI =\coins$ be an instance of \colkcenosc. Let $(O^*=\{o^*_1,\dots, o^*_k\}, R^*=\{r^*_1,\dots,r^*_k\})$ be an optimal solution to $\cI$ with cost $\opt:= ||R^*||$, and let $\Pi^* = \{\pi^*_1,\dots,\pi^*_k\}$ be the corresponding clustering of $C$. Let $Z^* \subseteq C$ be the set of outliers for $\Pi^*$. Without loss of generality, for example by relabeling centers in $O^*$, we assume that $o^*_i \in F_i$, for $i \in [k]$.
Finally, let $R=\{r_1,\dots,r_k\}$ be an $\epsilon$-approximation of $R^*$, obtained from~\Cref{thm:epsapxrad}.


\subsection{Analysis}\label{ss:analysis}
For the analysis, we make the following assumption on the guesses. 
All guesses are implemented via bounded branching; see~\Cref{ss:runtime}.
\begin{assumption}\label{assm:guess}
	The \guess keyword always returns the correct answer, if there exists. Otherwise, it returns \emph{unsuccessful}.
\end{assumption}

For analyzing the algorithms, we assume that~\Cref{assm:guess} holds. We start with setting up the notations required for the analysis.

\paragraph*{Notations.}
Consider an iteration $i \ge 1$ of the \textbf{while}-loop of~\Cref{algo:colksupalgo}.
\begin{itemize}
	\item Let $\Pi^*_{u,i}$ denote the unmarked optimal clusters at the beginning of iteration $i$.
	\item Let $C'_i$ denote the remaining clients at the beginning of iteration $i$, and let and $C'_{u,i} \subseteq C'_i$ denote the points of unmarked optimal clusters in $C'_i$ at the beginning of iteration $i$. Also, let $C'_{f,i}:=C'_i \setminus C'_{u,i}$, which we call \emph{free points} in iteration $i$. 
	\item Let $j_i \in [k]$ be the index of densest unmarked optimal cluster in $C'_i$ (picked by~\Cref{algo:colksupalgo} in line~\ref{algo:lin:den}).
\end{itemize}

Next, we need the following definition that is used by~\Cref{algo:settle cluster}.

\begin{definition}[nearby, good, and light ball]\label{def:allballs}
	Consider an iteration $i \geq 1$. 
	For $C''_i \subseteq C'_i$, a ball $B=\ball_{C''_i}(\hat{t},r_{j_i}), \hat{t} \in F$ is
	\begin{enumerate}
		\item[(i)] a \emph{nearby ball} to $\pi^*_{j_i}$ if $d(\hat{t},o^*_{j_i}) \le 2\cdot r_{j_i}$ and $\hat{t} \in F_{j_i}$.
		\item[(ii)] a \emph{good ball} for $\pi^*_{j_i}$ in $C'_i$ if $\hat{t} \in F_{j_i}$ and the number of free points in ${B}$ is at least the number of points of $\pi^*_{j_i}$  in $C'_i$, i.e., $|{B} \cap C'_{f,i}| \geq |\pi^*_{j_i} \cap C'_{i}|$.
		\item[(iii)] a \emph{dense ball} for iteration $i$ if ${B}$ contains strictly more than ${B}/2$ points from $C'_{u,i}$, i.e.,  $|{B} \cap C'_{u,i}| > |{B}|/2$; otherwise ${B}$ is a \emph{light ball} for iteration $i$, in which case $|{B} \cap C'_{u,i}| \le |{B}|/2$, or equivalently,  $|{B} \cap C'_{f,i}| \geq |{B}|/2$, i.e., at least half of its points are free.
	\end{enumerate}
\end{definition}


See~\Cref{fig:ballstri} for an illustration.
Next, we make the following simple observation.

\begin{observation}\label{obs:Bnotempty}
	Consider an iteration $i \ge 1$ of the \textbf{while}-loop of~\Cref{algo:colksupalgo}. First, the \guess of~\Cref{algo:colksupalgo} (in line~\ref{algo:lin:den}) is always successful.
	Second,  the set $\cB$ of balls constructed by~\Cref{algo:settle cluster} for iteration $i$ is not empty.
\end{observation}
\begin{proof}
	The first sentence follows due to the condition of the \textbf{while}-loop being satisfied. For the second sentence, consider $\pi^*_{j_i}$ and its center $o^*_{j_i}$. Since the algorithms never delete any facility from $F$, it holds that $o^*_{j_i} \in F_{j_i}$, and hence $\cB \neq \emptyset$. 
\end{proof}
The following key lemma provides structural properties of the ball set $\cB$ constructed by~\Cref{algo:settle cluster} that are crucial in obtaining the desired solution.
\begin{lemma}[Structural trichotomy]\label{lem:charging lemma}
	Consider an iteration $i \geq 1$ of the \textbf{while}-loop of~\Cref{algo:colksupalgo}. Let $\cB$ be the set of balls obtained by~\Cref{algo:settle cluster} at the end of its \textbf{while}-loop in iteration $i$. Then, at least one of the following is true.
	\begin{enumerate}
		\item[(i)] $\cB$ contains a good ball for $\pi^*_{j_i}$ in $C'_i$
		\item[(ii)] $\cB$ contains a nearby ball to $\pi^*_{j_i}$
		\item[(iii)] there exists $\hat{B}_1\neq \hat{B}_2 \in \cB$ and an unmarked optimal cluster $\pi^*_{\ell_i} \in \Pi^*_{u,i}$ such that
		\begin{itemize}
			\item both $\hat{B}_1, \hat{B}_2$ are light balls for iteration $i$,
			\item both $\hat{B}_1, \hat{B}_2$ are nearby balls for $\pi^*_{\ell_i}$, and
			\item $\hat{B}_1 \cup \hat{B}_2$ contain at least $|\pi^*_{j_i} \cap C'_i|$ free points, i.e., $|(\hat{B}_1 \cup \hat{B}_2)\cap C'_{f,i}| \ge |\pi^*_{j_i} \cap C'_i|$.
		\end{itemize}
	\end{enumerate}
\end{lemma}

Intuitively, density forces each constructed ball to draw its mass either from free points or from remaining optimal clusters, and a  counting argument ensures that one of the above configurations must arise.
We defer the proof of the structural lemma to the end of this section.
Now we show the correctness of~\Cref{algo:colksupalgo} using the structural lemma. First, we show that~\Cref{algo:colksupalgo} terminates quickly and successfully (without failing).

\begin{lemma}\label{lem:iterofalgo}
	The \textbf{while}-loop of~\Cref{algo:colksupalgo} runs at most $k$ times.
\end{lemma}
\begin{proof}
	First note that the \guess of~\Cref{algo:colksupalgo} is always successful due to~\Cref{obs:Bnotempty}.
	Now, we claim that in every iteration of the \textbf{while}-loop of~\Cref{algo:colksupalgo}, at least one unmarked optimal cluster is marked. Since there are $k$ unmarked optimal clusters in the beginning, this implies the lemma. For the claim, consider an iteration $i \geq 1$ of the \textbf{while}-loop, and consider the set $\cB$ of balls constructed by~\Cref{algo:settle cluster} in this iteration. Recall that $\pi^*_{j_i}$ is the unmarked optimal cluster considered by the algorithm in line~\ref{algo:lin:den}.
	Now, if $\cB$ contains a nearby ball to $\pi^*_{j_i}$ or a good ball for $\pi^*_{j_i}$ in $C'_i$, then the \guess in line~\ref{algo:sc:line:goodcase} is successful and~\Cref{algo:settle cluster} marks $\pi^*_{j_i}$ in line~\ref{algo:sc:lin:markj}, and hence we are done. Otherwise,~\Cref{lem:charging lemma} implies the existence of balls $\hat{B}_1, \hat{B_2}$ in $\cB$ such that the \guess in Line~\ref{algo:sc:lin:badcase} is successful. Furthermore, $o^*_\ell \in F_{\ell}$ is a candidate for $\hat{t}'_2$ since $d(o^*_\ell, \hat{t}_2) \le r^*_\ell + r_{j_i} \le  r_\ell + r_{j_i}$. Therefore,~\Cref{algo:settle cluster} marks $\pi^*_{j_i}$ and $\pi^*_\ell$ in line~\ref{algo:lin:markell}, finishing the proof of the claim.
\end{proof}

Given the above lemma, we make the following definition with respect to the execution of~\Cref{algo:settle cluster} for iteration $i \in [k]$. We say that iteration $i$ is in the \emph{easy case} if the \guess in line~\ref{algo:sc:line:goodcase} of~\Cref{algo:settle cluster} is successful, otherwise we say that it is in the \emph{hard case}.
Next, let $\cE= \cE_e \cup \cE_h$, where $\cE_e$ are the solutions corresponding to the easy case, and $\cE_h$ corresponds to the solutions to the hard case. Since~\Cref{algo:settle cluster} marks each optimal cluster exactly once, we have that $|\cE|=k$. Furthermore, since in the hard case, \Cref{algo:settle cluster} returns two solutions, we pair the solutions in $\cE_h$, and we say that the corresponding optimal clusters marked by the algorithm as \emph{settled} together.

\begin{lemma}\label{lem:feasE}
	The cost of the solution $\cE$, returned by~\Cref{algo:colksupalgo}, is at most $(3+\eps) \opt$. Furthermore, $\cE$ satisfies the fairness constraints of $\cI$.
\end{lemma}
\begin{proof}
	We first argue about the cost of $\cE$ for sum objective. For $\cE_e$, we have that the sum of radii of the solutions in $\cE_e$ is $\sum_{(t_j,3r_j) \in \cE_e} 3r_j = 3\cdot \sum_{(t_j,3r_j) \in \cE_e} r_j$. Now consider a pair $(t_j,r_j+2r_\ell), (t_\ell,2r_j+r_\ell)$ in $\cE_h$, we have that the sum of radii of this pair is $3r_j+3r_\ell$. Since the clusters marked in the easy case are disjoint from those marked in the hard case, we have that the sum of radii of the solution $\cE$ is $3\sum_{i \in [k]} r_i \le (3+\eps)\opt$, as $\{r_i\}$ is an $\eps$-approximation of the optimal radii $\{r^*_i\}$.
	Next, we argue that $\cE$ satisfies the fairness constraints of $\cI$. First note that a solution $(t_j,r_j) \in \cE_e$ corresponds to the marking of $\pi^*_j$, and hence in this case, we have that $t_j \in F_j$. Consider a pair $(t_j,r_j+2r_\ell), (t_\ell,2r_j+r_\ell)$ in $\cE_h$ corresponding to the marking of $\pi^*_j$ and $\pi^*_\ell$ (see~\Cref{fig:twolightballs} for an illustration). In this case, we have that $t_j \in F_j$ and $t_\ell \in F_\ell$ due to~\Cref{algo:settle cluster}. Therefore, this pair satisfies the fairness constraints of $F_j$ and $F_\ell$, finishing the proof.
	
	For the case when the objective is a monotone symmetric norm of the radii of the clusters, it is easy to see that the factor $(3+\eps)$ still holds.
\end{proof}

Due to the above lemma, it is now sufficient to show that  $\cE$ covers at least $|C|-z$ points of $C$. For the rest of the analysis, relabel the optimal clusters $\pi^*_{1},\dots, \pi^*_{k}$ in the order of their marking by the algorithm. 
Let $k' \in [k]$ be the total number of iterations performed by  the \textbf{while} loop of~\Cref{algo:colksupalgo}.
For every iteration $i \in [k']$  of the \textbf{while} loop of~\Cref{algo:colksupalgo}, let $\cE_i$ be the solution constructed by~\Cref{algo:colksupalgo} at the end of $i^{th}$ iteration. Moreover, let $\Pi^*_{i_m}=\{\pi^*_1,\dots, \pi^*_{i_m}\} \subseteq \Pi^*$ be the set of marked optimal clusters at the end of $i^{th}$ iteration. Note that $\cE_{k'}= \cE$ and $\Pi^*_{i_{k'}} = \Pi^*$.
The following lemma implies that $\cE$ covers at least as many points as covered by the optimal clusters.
\begin{lemma}\label{lem:covE}
	Consider an iteration $i \in [k']$ of  the \textbf{while} loop of~\Cref{algo:colksupalgo}. Then,   the number of points covered by the balls in solution $\cE_i$ is at least the number of points in $\Pi^*_{i_m}$, i.e.,  $|\cup_{(t,r) \in \cE_i} \ball(t,r)| \ge |\pi^*_1\cup\dots\cup\pi^*_{i_m}|$. 
\end{lemma}
\begin{proof}
	We prove the lemma using the following charging argument.
	We show, by induction on $i \in [k']$, that each point of $\Pi^*_{i_m}$ can be charged to a distinct point in the balls of $\cE_i$. 
	\paragraph*{Charging scheme.} We show that each point in $\pi^*_{i_m}$ can be charged to a distinct point in the balls of $\cE_i$ using exactly one of the two charging rules. In the first charging rule, which we call \emph{self-charging} rule, a point in $\pi^*_{i_m}$ is charged to itself, if it is covered by the balls of $\cE_i$
	Note that each point in $\pi^*_{i_m}$  can be charged at most once by this charging rule, since it is contained exactly in one optimal cluster.
	In the second charging rule, called \emph{freepoints-charging} rule, a point in $\pi^*_{i_m}$ is charged to a distinct free point (of iteration $i$) in  $\cE_i$, provided such a point exists in $\cE_i$. Recall that a free point of iteration $i$ is a point $p \in C'_i$ that does not belong to any of the remaining unmarked optimal clusters, i.e., $p \in C'_i$ and $p \notin \pi^*_{i_{m}+1}\cup\dots\cup\pi^*_{i_k}$.
	Note that future optimal clusters can not charge these points again using self-charging rule since they do not belong to any of the remaining unmarked optimal cluster. Additionally, no future optimal cluster can charge them using freepoints-charging rule since they are deleted at the end of the iteration.
	
	Now  we are ready for the inductive argument using the above charging scheme.
	For the base case, consider the first iteration, $i=1$.
	If it is in the easy case, then we have $\Pi^*_{1_m}=\{\pi^*_1\}$, and hence    let $(\hat{t},3r_1)$ be the solution present in $\cE_1$. Then, $\hat{B}=\ball(\hat{t},r_1) \in \cB$ is either a nearby ball to $\pi^*_{1}$ or a good ball for $\pi^*_{1}$ in $C'_1$. For the former case, note that $\pi^*_{1} \subseteq \ball(\hat{t},3r_1)$ due to triangle inequality, and hence  use self-charging rule to charge the points of $\pi^*_{1}$ to themselves. 
	For the latter case, note that $C'_1=C$, and hence $C'_{f,1}=Z^*$, the set of outliers point for $\Pi^*$. Since $\hat{B}$ is a good ball for $\pi^*_{1}$ in $C'_1=C$, it contains at least $|\pi^*_{1} \cap C|=|\pi^*_{1}|$ points from $Z^*$. Hence, charge each point of $\pi^*_{1}$ to a distinct outlier point in $\hat{B}$ using the freepoint charging rule.
	Therefore, $|\ball(\hat{t},3r_1)| \ge |\pi^*_1|$ in the easy case.
	Now suppose the iteration is in the hard case, then note that $\Pi^*_{1_m}= \{\pi^*_1,\pi^*_2\}$ as $\pi^*_1$ and $\pi^*_2$ are settled together by~\Cref{algo:settle cluster}.
	Let $\ball(\hat{t}_1,r_1)$ and $\ball(\hat{t}_2,r_2)$ be the balls obtained by the algorithm from~\Cref{lem:charging lemma}.
	In this case, $\cE_1$ contains $(\hat{t}_1,r_1+2r_2)$ and $(\hat{t}'_2,2r_1+r_2)$. 
	Since $\ball(\hat{t}_1,r_1)$ is a nearby ball to $\pi^*_2$, we have that  $\pi^*_2 \subseteq \ball(\hat{t}_1,r_1+2r_2)$, and hence, charge each point of $\pi^*_2$ to itself using self-charging rule. Furthermore, since both $\ball(\hat{t}_1,r_1)$  and $\ball(\hat{t}_2,r_2)$ are light balls for iteration $1$, we have that they together contain at least $|\pi^*_1|$ freepoints, since each of them has at least half of its points as freepoints. Therefore, charge each point of $\pi^*_1$ to a distinct free point in $\ball(\hat{t}_1,r_1) \cup \ball(\hat{t}_2,r_2) \subseteq \ball(\hat{t}_1,r_1+2r_2) \cup \ball(\hat{t}'_2,r_1+2r_2)$   using freepoint-charging rule, since $d(\hat{t}'_2,\hat{t}_2) \le r_1+r_2$ (see~\Cref{fig:twolightballs}).
	Hence, the hypothesis holds for the first iteration.
	
	
	

	Now for the induction hypothesis, suppose  each point in $\Pi^*_{(i-1)_m}$ has been charged to a distinct point in the balls of $\cE_{i-1}$. Now, suppose iteration $i$ is in the easy case, and hence, we have $\Pi^*_{i_m} = \Pi^*_{(i-1)_m} \cup \{\pi^*_{i_m}\}$.
	Therefore, let $(\hat{t},3r_{i_m})$ be the solution added to $\cE_{i-1}$ to obtain $\cE_i$.
	Moreover, $\hat{B}=\ball(\hat{t},r_{i_m}) \in \cB$ is either a nearby ball to $\pi^*_{i_m}$ or a good ball for $\pi^*_{i_m}$ in $C'_i$. For the former case, note that $\pi^*_{i_m} \subseteq \ball(\hat{t},3r_{i_m})$ due to triangle inequality, and hence use the self-charging rule to charge the points of $\pi^*_{i_m}$ to themselves.
	For the latter case, since $\hat{B}$ is a good ball for $\pi^*_{i_m}$ in $C'_i$, it contains at least $|\pi^*_{i_m} \cap C'_i|$ free points. Therefore, use self-charging rule to charge each point of $\pi^*_{i_m}$ contained in  the balls of $\cE_{i-1}$ to itself, and
	use freepoint-charging rule to charge each point of $\pi^*_{i_m} \cap C'_i$ to a distinct free point in $\hat{B}$. 
	Now suppose iteration $i$ is in the hard case, and hence we have $\Pi^*_{i_m} = \Pi^*_{(i-1)_m} \cup \{ \pi^*_{i_{m-1}},  \pi^*_{i_m}$\}. 
	In this case, the algorithm adds $(\hat{t}_1,r_{i_{m-1}}+2r_{i_m})$ and $(\hat{t}'_2,2r_{i_{m-1}}+r_{i_m})$. 
	See~\Cref{fig:twolightballs} for an illustration.
	Since $\ball(\hat{t}_1,r_{i_{m-1}})$ is a nearby ball to $\pi^*_{i_m}$, we have that  $\pi^*_{i_m} \subseteq \ball(\hat{t}_1,r_{i_{m-1}}+2r_{i_m})$, and hence, charge each point of $\pi^*_{i_m}$ to itself using self-charging rule. Furthermore, since both $\ball(\hat{t}_1,r_{i_{m-1}})$  and $\ball(\hat{t}_2,r_{i_m})$ are light balls for iteration $i$, we have that they together contain at least $|\pi^*_{i_{m-1}} \cap C'_i|$ freepoints from~\Cref{lem:charging lemma}. 
	Therefore, use self-charging rule to charge each point of $\pi^*_{i_{m-1}}$ contained in  the balls of $\cE_{i-1}$ to itself, and
	use freepoint-charging rule to charge each point of $\pi^*_{i_{m-1}} \cap C'_i$ to a distinct free point in $\ball(\hat{t}_1,r_{i_{m-1}}) \cup \ball(\hat{t}_2,r_{i_{m}}) \subseteq \ball(\hat{t}_1,r_{i_{m-1}}+2r_{i_{m}}) \cup \ball(\hat{t}'_2,r_{i_{m-1}}+2r_{i_m})$   using freepoint-charging rule, since $d(\hat{t}'_2,\hat{t}_2) \le r_{i_{m-1}}+r_{i_m}$.
\end{proof}

Now we prove the structural trichotomy.

\begin{proof}[Proof of~\Cref{lem:charging lemma}]
	First note that $\cB$ is non-empty due to~\Cref{obs:Bnotempty}. Now consider the case when $|\cB| < 4k$. Then, we claim that $\cB$ contains a nearby ball for $\pi^*_{j_i}$, implying (ii).
	Consider the \textbf{while}-loop of~\Cref{algo:settle cluster}, and since $|\cB| < 4k$, $F''_{j_i}$ is  empty  at the end of this loop. Since the \textbf{while}-loop only deletes the facilities in $F''_{j_i}$ after picking them as a center of a ball in $\cB$, this means that there is ball $\hat{B} \in \cB$ centered at $o^*_{j_i}$. This implies (ii) as $\hat{B}$ is a nearby ball for $\pi^*_{j_i}$.
	
	Now consider the interesting case where $\cB$ has exactly $4k$ balls but none of them is a nearby ball for $\pi^*_{j_i}$ . In this case, we will show that either $(i)$ or $(iii)$ is true.
	We start with the following simple observation, whose proof follows from the fact that the balls in $\cB$ are pairwise disjoint in $C'_i$ with radius $r_{j_i} \ge r^*_{j_i}$ and none of them contain a point from $\pi^*_{j_i}$.
	\begin{observation}\label{obs:pointsofB}
		If $\cB$ does not contain a nearby ball for $\pi^*_{j_i}$, then each ball in $\cB$ contain at least $t_{j_i}:=|\pi^*_{j_i} \cap C'_i|$ many distinct points from $C'_i$.
	\end{observation}
	
	Let $\cL\subseteq \cB$ be the set of light balls for iteration $i$, and  let $\cD = \cB \setminus \cL$, be the set of dense balls for iteration $i$. 
	We claim that $|\cD| < 2k$, implying $|\cL| \ge 2k$. For contradiction suppose $|\cD| \ge 2k$. We  will show that  $\cB$ contains a nearby ball for $\pi^*_{j_i}$, contradicting the assumption for this case. Let $u_i$ be the number of points of unmarked optimal clusters in $C'_i$.  Since $\pi^*_{j_i}$ is the densest unmarked optimal cluster in $C'_i$, we have that $u_i \le t_{j_i}\cdot k$, where $t_{j_i}=|\pi^*_{j_i} \cap C'_i|$. On the other hand, since none of the balls in $\cD$ contain a point from $\pi^*_{j_i}$, we have that each ball in $\cD$ is disjoint from $\pi^*_{j_i} \cap C'_i$.  Therefore, from~\Cref{obs:pointsofB}, every ball in $\cD$ contains at least $t_{j_i}$ many distinct points from $C'_i$. Moreover, since each ball in $\cD$ is a dense ball for iteration $i$,  each ball in $\cD$ contains more than $\frac{t_{j_i}}{2}$ many points from the remaining unmarked optimal clusters. Therefore, the total number of points in $\cD$ that belong to the remaining unmarked optimal clusters is $>\frac{t_{j_i}}{2} \cdot |\cD| = t_{j_i}\cdot k$, contradicting the upper bound on $u_i$.
	
	Now consider $\cL, |\cL|\ge 2k$. If $\cL$ contains a good ball for $\pi^*_{j_i}$ then we are done since (i) is true in this case. Hence, assume that no ball in $\cL$ is a good ball for $\pi^*_{j_i}$. This means that every ball in $\cL$ has a point that belongs to one of the remaining unmarked optimal clusters (except $\pi^*_{j_i}$). Since $|\cL| \ge 2k \ge k+1$, there exists a unmarked optimal cluster $\pi^*_{\ell_i}$ and balls $\hat{B}_1 \neq \hat{B}_2 \in \cL$ such that $\hat{B}_1$ and $\hat{B}_2$ each contain a distinct point from $\pi^*_{\ell_i} \cap C'_i$. This means that $\hat{B}_1$ and $\hat{B}_2$ are nearby balls for $\pi^*_{\ell_i}$. Moreover, since each of $\hat{B}_1,\hat{B}_2$ is a light ball for iteration $i$, we have that each one of them has half the points that are free points, i.e, $|\hat{B}_g \cap C'_{f,i}| \geq \frac{|\hat{B}_g|}{2}$, for $g \in [2]$. As these balls are pairwise disjoint in $C'_i$, we have that
	\[
	|(\hat{B}_1 \cup \hat{B}_2) \cap C'_{f,i}| = |\hat{B}_1 \cap C'_{f,i}| + |\hat{B}_2 \cap C'_{f,i}| \geq \frac{|\hat{B}_1| + |\hat{B}_2|}{2} \geq t_{j_i} \geq |\pi^*_{\ell_i} \cap C'_i|,
	\]
	where the second last inequality follows due $|\hat{B}_1| + |\hat{B}_2| \geq 2t_{j_i}$ from~\Cref{obs:pointsofB}, and the last inequality follows from the fact that $\pi^*_{j_i}$ is the densest unmarked optimal cluster in $C'_i$. This implies (iii).
\end{proof}

\subsection{Running time}\label{ss:runtime}
We next bound the running time of our algorithms without using the \guess keyword. Specifically, we show how to replace each \guess keyword using a bounded branching program. First, assume that~\Cref{algo:colksupalgo} is given an $\eps$-approximation of $R^*$ as an input  and hence both~\Cref{algo:colksupalgo} and~\Cref{algo:settle cluster} know these radii. 
With this assumption, we next explain how to replace other \emph{guess}es. 

Each iteration of~\Cref{algo:colksupalgo} has to guess the index of the densest unmarked optimal cluster, which results into $k$ branches.
Next, at every call of~\Cref{algo:settle cluster}, the structural trichotomy implies three possible cases.
For the nearby-ball and good-ball cases, we have to guess which of the constructed balls satisfies the corresponding property, which can be done by enumerating $O(k)$ choices for each case.
For the two-light-balls case, we have to guess the relevant pair of balls and the corresponding unmarked cluster index, which requires enumerating $O(k^3)$ choices.
Thus, each call of~\Cref{algo:settle cluster} branches into at most $O(k)+ O(k^3)=O(k^3)$ choices.

First, note that with $\eps$-approximation of $R^*$ and the structural trichotomy, \Cref{algo:settle cluster} never fails.
As \Cref{algo:colksupalgo} runs at most $k$ iterations
since, in each iteration, \Cref{algo:settle cluster},  settles at least one optimal cluster, the total number of nodes in this branching program is bounded by $(k\cdot O(k^3))^k = k^{O(k)} = 2^{O(k\log k)}$.
Moreover, each node in the branching program performs only polynomial work in $n$.
Finally, we remove the assumption that the algorithms know $\eps$-approximation of $R^*$ as follows. From~\Cref{thm:epsapxrad}, we have a list $\cL$ of candidate radii that contains an $\eps$-approximation of $R^*$. For each candidate entry in $\cL$, we run the branching version of~\Cref{algo:colksupalgo} as mentioned above, and return the candidate radii with minimum norm for which~\Cref{algo:colksupalgo} does not fail. Therefore, the overall running time of~\Cref{algo:colksupalgo} is $2^{O(k \log (k/\eps))}\poly(n)$.






\section{Extension to Lower and Upper Bound Requirements}\label{sec:extn}

We extend our results to the \emph{fair-range SoR with outliers} problem, where the point set $X$ is partitioned into demographic groups $G_1,\dots,G_t$, and each group $G_i$ is associated with a lower bound $\ell_i$ and an upper bound $u_i$.
The goal is  to find $k$-sized set $S=\{s_1,\dots,s_k\} \subseteq \points$ and a radii set $R=\{r_1,\dots,r_k\} \in \mathbb{R}_{\ge 0}$
such that  $|(\cup_{i \in [k]} \ball(s_i,r_i))| \ge n-z$ and \[
\ell_i \le |S \cap G_i| \le u_i \qquad \text{for all } i \in [t],
\]
so as to minimize $\sum_{i \in [k]} r_i$. 
We assume feasibility, i.e., $\sum_{i=1}^t \ell_i \le k \le \sum_{i=1}^t u_i$.

\medskip

The high level idea is to reduce the fair-range constraints into a collection of unit-requirement groups, with both equality and upper-bound constraints.
In more detail, for each group $G_i$, we first enforce the lower bound by creating $\ell_i$ \emph{equality groups}, each of which is a copy of $G_i$ and must contribute exactly one center.
This guarantees that any feasible solution selects at least $\ell_i$ centers from $G_i$.
After these mandatory selections, the group $G_i$ has a remaining allowance of $u_i-\ell_i$ additional centers.
To capture this flexibility, we create $u_i-\ell_i$ additional copies of $G_i$, referred to as \emph{slack groups}, each with unit requirement allowing the selection of at most one center.
Selecting centers from these slack groups permits additional representatives from $G_i$ but caps the total number of selections at $u_i$.
Thus, the equality groups enforce the lower bounds, while the slack groups encode the upper bounds, and every group in the resulting instance has a unit requirement.

We now apply the same color-coding technique as in Section~\ref{ss:reduc} to reduce the instance to all unit-requirement groups into exactly $k$ color classes, obtaining a colorful $k$-SoR instance with outliers.
Since a feasible colorful solution selects exactly one center from each color class, it corresponds to selecting exactly one center from each equality and slack group, thereby satisfying all lower and upper bound constraints.
Next, we use~\Cref{algo:colksupalgo} to obtain a $(3+\eps)$-approximate solution for this colorful instance in \fpt\ time, which in turn, is a $(3+\eps)$-approximate solution for the fair-range instance.
The correctness and approximation guarantee follow identically from the analysis of the upper-bound case.
We therefore obtain the following extension of our main result.
\fairangethm*

\subsubsection*{Acknowledgment}
We thank Tanmay Inamdar for his valuable suggestions on an earlier version of this paper, which considered only the $k$-center objective.

\bibliography{refs}
\bibliographystyle{plain}

\end{document}